\documentclass[letterpaper, 11 pt, compsocconf]{IEEEtran} 

\newcommand{\subparagraph}{}
\def\BibTeX{{\rm B\kern-.05em{\sc i\kern-.025em b}\kern-.08emT\kern-.1667em\lower.7ex\hbox{E}\kern-.125emX}}

\usepackage[utf8]{inputenc}
\usepackage{todonotes}
\usepackage{url}
\usepackage{caption,subcaption}  
\usepackage{booktabs}
\usepackage{pgfplots}
\pgfplotsset{width=10cm,compat=1.9}
\usepackage{adjustbox}
\usepackage{mathtools}

\DeclarePairedDelimiter\floor{\lfloor}{\rfloor}
\usepackage{amsthm}
\usepackage{breakurl}

\usepackage[small,compact]{titlesec}
\usepackage{wrapfig}
\usepackage[font=small,labelfont=bf]{caption} 
\newtheorem{theorem}{Theorem}
\newtheorem{lemma}[theorem]{Lemma}

\usepackage{epsfig}
\usepackage[linesnumbered,algoruled,boxed,lined]{algorithm2e}

\makeatletter 
\g@addto@macro{\@algocf@init}{\SetKwInOut{Parameter}{Parameters}} 
\makeatother

\usepackage{comment}
\usepackage{amsmath}
\usepackage{amsfonts}
\usepackage{amssymb}
\usepackage{amsthm}
\usepackage{pgfplots}
\usepackage{tikz}
\usepgfplotslibrary{groupplots}
\pgfplotsset{compat=newest}

\usepackage{mathtools}

\begin{document}

\sloppy

\title{Proteus: A Scalable BFT Consensus Protocol for Blockchains}

\author{
\IEEEauthorblockN{Mohammad M. Jalalzai\textsuperscript{1,2}, Costas Busch\textsuperscript{1}, Golden G. Richard III\textsuperscript{1,2}}\\
\IEEEauthorblockA{
{\textsuperscript{1}Computer Science and Engineering Division}\\ {\textsuperscript{2} Center for Computation and Technology}\\
Louisiana State University
Baton Rouge, Louisiana, USA\\
}
Email: \url{mjalal7@lsu.edu}, \url{busch@csc.lsu.edu}, \url{golden@cct.lsu.edu}

}

\pagestyle{plain}
\maketitle
\thispagestyle{plain}

\begin{abstract}

Byzantine Fault Tolerant (BFT) consensus exhibits higher throughput in comparison to Proof of Work (PoW) in blockchains, but BFT-based protocols suffer from scalability problems with respect to the number of replicas in the network. The main reason for this limitation is the quadratic message complexity of BFT protocols. Previously, proposed solutions improve BFT performance for normal operation, but will fall back to quadratic message complexity once the protocol observes a certain number of failures. This makes the protocol performance unpredictable as it is not guaranteed that the network will face a a certain number of failures. 
As a result, such protocols are only scalable when conditions are favorable (i.e., the number of failures are less than a given threshold). To address this issue we propose {\em Proteus}, a new BFT-based consensus protocol which elects a subset of nodes $c$ as a root committee. Proteus guarantees stable performance, regardless of the number of failures in the network and it improves on the quadratic message complexity of typical BFT-based protocols to $O(cn)$, where $c<<n$, for large $n$. 

Thus, message complexity remains small and less than quadratic when $c$ is asymptotically smaller than $n$, and this helps the protocol to provide stable performance even during the view change process (change of root committee). Our view change process is different than typical BFT protocols as it replaces the whole root committee compared to replacing a single primary in other protocols. We deployed and tested our protocol on $200$ Amazon $EC2$ instances, with two different baseline BFT protocols (PBFT and Bchain) for comparison. In these tests, our protocol outperformed the baselines by more than $2\times$ in terms of throughput as well as latency. 
\end{abstract}





\section{Introduction} \label{Introduction}

A Blockchain is a distributed ledger in which blocks of transactions are stored in sequential order. Participants in blockchain networks use consensus through a peer-to-peer network to agree on each block of transactions. Each block contains an ordered set of transactions and a link (hash) to the previous block in the chain. Traversal of the hashes to previous blocks allows us to move through the history of transactions. Immutability for blocks ``buried'' in the blockchain is probabilistically guaranteed, since modifying a block invalidates the hashes of all newer blocks in the chain. Recently there has been a lot of work to address the scalability (number of replicas in the network), throughput (in terms number of transactions per second) and latency (time to insert the block of a transaction in the blockchain) issues of blockchains \cite{Jala1807:Window,DBLP:journals/corr/LiuLKA16a, Guerraoui:2010:NBP:1755913.1755950,SBFT,194906}. But increasing the number of replicas (participants) in the network can have negative effects on latency and throughput due to the increase in the number of messages being exchanged and processed within the network.

Protocols based on Proof of Work (PoW) are suffering from high latency and low throughput, but have high scalability. The main bottleneck of performance for PoW protocols lies in solving the cryptographic puzzle before proposing a block for inclusion in the chain. Bitcoin is an example of a PoW based protocol that has been shown to accommodate thousands of replicas, but its throughput is 6-10 transactions per second and it takes an average of 10 minutes to generate a new block~\cite{Nakamoto_bitcoin:a,DBLP:conf/ifip114/Vukolic15}. Bitcoin's PoW is very CPU-intensive and is responsible for the high consumption of electricity in Bitcoin (comparable to the entire electricity consumption of Ireland) \cite{6912770,IsBitCoinDecentralized,Sok}.
Ethereum is another well known blockchain ledger that uses PoW \cite{Ethereum-EIP-150}. While Ethereum's approach to PoW is somewhat different than Bitcoin's, primarily in an attempt to stop ASIC-enhanced mining, it suffers from some of the same drawbacks as Bitcoin.
In general, PoW-based protocols do not appear to be suitable for applications that require low latency and high transaction throughput. Furthermore, according to Luu {\em et al.} \cite{Smart-pools} $95\%$ of Bitcoin and $80\%$ of Ethereum networks' mining power resides within less than ten and six mining pools respectively, making the networks susceptible to possible $51\%$ attacks.  This is because these pools operate in a centralized fashion, with pool owners directly controlling the work of individual miners.
In PoW protocols there is also the risk of multiple forks that can result in double-spends.
Thus, even after a transaction is committed, clients have to wait a specific number of blocks to make sure the transaction is ``finalized'' in the longest fork. For Bitcoin, this usually amounts to six blocks (requiring 60 minutes)~\cite{Nakamoto_bitcoin:a} and for
Ethereum, this threshold is generally 37 blocks (due to a much faster block commit rate) \cite{wust2016securityofblockchaintech}.

On the other hand, Byzantine Fault Tolerant (BFT) protocols \cite{Lamport:1984:UTI:2993.2994}  are able to achieve block consensus in the presence of Bytanzine (malicious) replicas with the exchange of a few rounds of messages. Byzantine replicas may fail in arbitrary ways by sending malicious messages, crashing, or coordinating malicious attacks, in order to prevent the network from reaching consensus.  BFT protocols have been shown to  process thousands of requests per second\cite{Kotla:2008:ZSB:1400214.1400236,Guerraoui:2010:NBP:1755913.1755950},
and are known to have higher throughput than PoW-based protocols. 
BFT-based protocols typically use a single replica as a {\em primary} to serialize requests/blocks during each epoch of creating blocks. Once consensus is achieved on the request proposed by the primary, the block will be added to the chain at the end of the epoch. If more than one third of replicas observe malicious behavior by the primary replica, a view change will be triggered. In the view change process the failed/malicious primary will be replaced with a new primary. 
Another attraction of BFT protocols is their their finality property. Finality means that once a block is committed it will \emph{never} be revoked. That means BFT protocols do not develop forks during the consensus process, thus once a transaction and its block are committed, the application layer can use the result without waiting for further confirmation.

The scalability of BFT-based protocols is a major concern. One of the main factors negatively affecting the scalability and performance of BFT-based protocols is that they require $n \times n$ broadcast for $n$ replicas (quadratic message complexity) \cite{Castro:1999:PBF:296806.296824,DBLP:conf/ifip114/Vukolic15,Luu:2016:SSP:2976749.2978389}.
This high communication overhead is to guarantee that consensus will always be reached 
even after under Byzantine failures. Typically the protocols guarantee that agreement will be reached if the total number of Byzantine nodes is less than a third of the total number of nodes.
However, the BFT protocols usually do not distinguish the cases where there are failures or not,
resulting in high message overhead in any case.

In practice, the reliability or fault tolerance of a protocol depends on the actual number of faults $(f)$ it can tolerate while providing desired throughput. Since the number of faults tolerated by $BFT$ is bounded by $f<n/3$ \cite{Fischer:1985:IDC:3149.214121}, the other way to increase $f$ is to increase $n$, in other words, to design a scalable protocol.

\subsection{Related Work}
Jalalzai {\em et al.} \cite{Jala1807:Window} have recently presented the Musch BFT protocol that addresses BFT scalability where the message complexity has been reduced to $O(f'n + n)$, where $f'$ is the actual number of Byzantine nodes ($f' < n\big /3$). The message complexity for small $f'$ is linear in $n$, making their protocol scalable. SBFT \cite{SBFT} is another protocol that uses sets of collectors called $c$ and $e$, which are randomly selected and collect signatures for the prepare and commit phases of consensus for the block proposed by the primary. These collectors help to avoid all-to-all broadcast. During optimistic execution, SBFT achieves $O(c n)$ message complexity. But it switches to plain PBFT if optimistic execution fails (i.e., if the collectors fail). This causes severe degradation in SBFT performance and results in $O(n^2)$ message complexity. On the other hand, our protocol never switches to PBFT.

FastBFT \cite{DBLP:journals/corr/LiuLKA16a} also tries to address the scalability and performance issues. In FastBFT the replicas are arranged in a tree structure where the primary node acts as the root of the tree. The message complexity in FastBFT in optimistic execution is $O(n \log n)$, but can switch to $O(n^2)$ in case of failures. Additionally, during optimistic execution signature  aggregation is done over the tree and each tree node is responsible for collecting signatures from its children. The improvement in message complexity and reduction in signature size occurs with cost of latency, as the critical path length grows to $O(log n)$.
 
In chain-based BFT protocols \cite{Guerraoui:2010:NBP:1755913.1755950,10.1007/978-3-319-14472-6_7}, the nodes in the network are arranged serially in a chain order. Message complexity during normal execution is $O(n)$ but can grow to $O(n^2)$ in the worst case (a view change of primary). These protocols have shown high throughput but the latency is proportional to the length of the critical path, which reaches $O(n)$.

All of the aforementioned solutions switch to $O(n^2)$ message complexity once a certain number of failures are detected in the network (this threshold varies among protocols). Thus, reaching the failure threshold causes the protocol to switch to broadcast mode (fallback mode), resulting in unusable performance for large-scale systems. As a result, these protocols only provide desired performance when failure thresholds have not been met. 

Hot-stuff \cite{Hot-stuff} is another BFT-based protocol that improves the view change message complexity of PBFT by a factor of $O(n)$ (dropping view change message complexity from $O(n^3)$ to $O(n^2)$. But during normal execution its message complexity matches that of PBFT ($O(n^2)$). It also relaxes the BFT finality condition to merge the prepare phase of the next proposal with the current commit phase. Tendermint \cite{tendermint} combines BFT and Proof-of-Stack (Pos) into a single protocol. An important aspect of Tenderment is the use of rotating leader election, where after each block, a new leader is selected deterministically. More stake for a validator is translated into more voting power and 
 proportionally more times to be selected as a leader.
But it suffers liveness issue due to conflicting proposals. \cite{BlockchaininDWild}.

\subsection{Contributions}

Our proposed protocol's performance is not bounded by any threshold on the number of failures. In other words, protocol performance is not affected by the number of failures detected in the network and remains constant (guarantees stable performance) during normal execution. This is a very important and strong characteristic that enables the protocol to provide consistent performance guarantees. Additionally, our protocol provides constant latency in terms of critical path length, which is the number of one-way messages from block proposal to completion of consensus. 

These improvements are achieved by randomly selecting $c$ number of replicas from a large set of replicas with total size $n$, where $n$ are regular replicas.
Typically $c$ is small and with very large $n$, we will have $c<<n$. Our algorithm can tolerate up to $f < n/3$ failures. In each epoch, the BFT consensus is first executed among the $c$ replicas of the root committee, instead of the overall $n$ replicas.
The block proposed by the root committee will then be validated by the $n-f$ correct replicas.
Since the main BFT process will be executed within root committee of size $c$, which is much smaller than $n$, this gives us improved performance.
The message complexity in our protocol is $O(c^2 + cn)$ for normal as well as view change mode, 
and when $n$ is large, where $c<<n$, it becomes $O(cn)$.
Through experimental evaluation we also show that our protocol
outperforms both PBFT and Bchain in terms of throughput and latency.

Our protocol can tolerate up to $2c/3$ Byzantine failures in the root committee. 
Therefore, the root committee is more resilient than the typical case which tolerates less than $c/3$ nodes failing.
The $n$ regular replicas keep an eye on failures of the $c$ root committee replicas that generate a block. In case $2c/3$ or more nodes fail, a view change will occur and another root committee is selected. This is another unique aspect of our protocol in that 
a view change causes the root committee to be replaced, whereas other BFT-based protocols only replace the primary node.

\subsection*{Paper Outline}
The paper is organized as follows.  
In Section \ref{System Model} we give our system model,
in Section \ref{Protocol} we present the Proteus protocol.
Its formal analysis appears in Section  \ref{Proof of correctness}
and the message complexity analysis in Section \ref{Message Complexity}.
Section \ref{section:experiments} contains the experimental analysis.
We conclude our work in Section \ref{Conclusions}.

\section{System Model}
\label{System Model}
   We consider the Byzantine fault model for the replicas.
    Under this model, servers and  clients may deviate from their normal behavior in arbitrary ways, which includes hardware failures, software bugs, and other malicious behavior. Our protocol can tolerate up to $f$ Byzantine replicas where the total number of replicas in the network is $n$ such that $n=3f+1$. The size of the root committee $c$ can be chosen based on a security guarantee requirement ($P_f$) for the committee executing the BFT consensus. 
   This  model also assumes that replicas will not be able to break collision resistant hashes, encryptions and signatures. We assume that all messages sent by the replicas are signed.
   To ensure liveness in the asynchronous network we use a timeout to place an upper bound on the block generation time. 




\begin{figure}
    \centering

\begin{tikzpicture}[scale=0.77]
\begin{semilogyaxis}[title={\textbf{Probability of failure to root committee size}},
 xlabel={size of c}, ylabel={Probability of failure},
  ytick={0,1e-16 ,1e-13 ,1e-10,1e-7, 1e-4,1e-1},
  enlargelimits=0.15,ymin=0,ymax=1e-1,
    legend style={at={(0.8,0.95  )},
      anchor=north},
 ]
 \addplot plot coordinates {
 (10,  0.0137)
 (20, 0.000172)
 (30, 5.260839926258207e-07)
 (40, 2.1e-9)

 };
\addplot plot coordinates {
 (10, 0.016)
 (20, 0.00036030388980864454)
 (30, 4.5e-06)
 (40, 1.9e-07)
 (50, 3.596244442968402e-10)
 (60, 7.5e-14)
 };
\addplot plot coordinates {

(10, 0.0164)
 (20, 0.0004273235534911956)
 (30, 4.945542105581939e-06)
 (40, 5.3314648793937834e-07)
 (50, 3.57e-9)
 (60, 2.251310238560559e-12)
 (70, 0)
 };
 \legend{$n=100$\\$n=150$\\$n=200$\\}
 \end{semilogyaxis}

 \end{tikzpicture}
\caption{Caption: Root committee size vs failure probability with different $n$}
    \label{fig:Failure probability for root committee}
\end{figure}
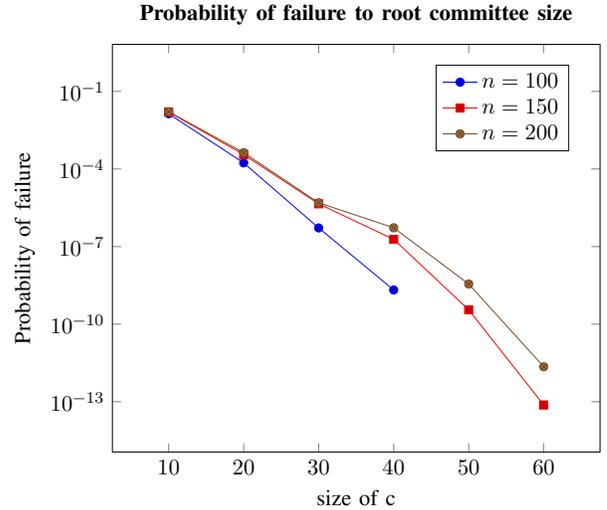

\section{Protocol}
\label{Protocol}

\label{protocol}
In Proteus, we have moved the burden of consensus to the root committee with size $c$. 
The root committee is running a customized $BFT$ based algorithm
(Algorithm \ref{Algorithm:Customized BFT}).
Regular replicas that are not members of the root committee simply verify the result of consensus, i.e. the proposed block. If $2c/3 + 1$ root committee replicas agree on the block proposed by the root committee, and in total $2f+1$ (root committee plus regular replicas) agree on the proposed block, then the block is committed and will be added to the blockchain. It should be noted that our protocol can tolerate up to two thirds ($2c/3$) Byzantine replicas in the root committee. Upon failure of the BFT protocol in root committee to propose a block, the whole root committee is replaced by the view change process.
Normal execution of the protocol can be summarized as:
\begin{enumerate}
\item The BFT protocol in the root committee successfully generates a block (collects more than two thirds {\em prepare} as well as {\em commit} messages for the block) and proposes it to regular replicas through broadcast. 
\item Upon receipt of a block, regular replicas verify the block against its history and check if the block is signed by $2c/3+1$ 
replicas from the root committee. 
\item If the block is valid, each regular replica signs the block and sends back the signature to the root committee.
\item Upon receipt of $2f+1$ signatures from regular replicas as well as root committee members, each root committee member commits the block and broadcasts the proof of acceptance ($2f+1$ signatures) of the block to regular replicas.
\item Upon receipt of proof, each regular replica commits the block, which is permanently added to the local history.
\end{enumerate}

\subsection{Selecting root committee members}
Members of the root committee are chosen randomly from the total number of replicas ($n$).
Suppose that out of the $n$ nodes $f$ are adversaries
(bad nodes) such that $f < n/3$.
The size $c$ is a predetermined number 
that specifies the average size of root committee.
Let $V$ be the set of nodes in the network with $|V| = n$.
We can write $V = A \cup B$,
where $A$ is the set of good nodes and $B$ 
is the set of bad nodes (adversaries),
such that $|A| = n -f$ and $|B| = f$.

Let $C$ denote the root-committee, such that $|C| = c$. 
We assume that $C$ is formed by randomly and uniformly picking 
a set of $c$ nodes out of $n$.
Therefore, the number of possible ways to pick any specific set 
of $c$ nodes is $\binom{n}{c}$.
The probability to pick exactly $a$ good nodes and $b$ bad nodes in $C$, 
such that $a + b = c$, is:
\begin{equation}
\label{eqn:basic1}
\frac{{\binom{n-f}{a}}{\binom{f}{b}}}{\binom{n}{c}}.
\end{equation}
Thus, the probability of having more than $2\cdot c/3$ bad nodes $(b)$ in root committee will be, the sum of all probabilities from $2c/3+1$ to $c$:
\begin{equation}
\label{eqn:basic1}
P_f=\begin{cases}
\sum_{b=2c/3+1}^{c}\frac{{\binom{n-f}{a}}{\binom{f}{b}}}{\binom{n}{c}},& \text{if } f > 2\cdot c/3\\
0,              & \text{otherwise.}
\end{cases}
\end{equation}


The safety of our algorithm is independent of the number of malicious/faulty replicas in the root committee. But for liveness we need at least one honest/correct replica to be a member of a $2c/3+1$ quorum that generates the block. We will discuss this in more detail in Section \ref{protocol}.

For the root committee to operate, it is randomly chosen from a set of $n$ replicas and the size $c$ is such that the failure probability having less than $c/3$ honest replicas is negligible as shown in Figure \ref{fig:Failure probability for root committee}. If the root committee fails to generate a valid block, it is replaced by another randomly chosen committee. Replacing the root committee with a newly chosen committee is called a view change. In our protocol, 
a view change is different from a typical BFT protocol view change, where only the primary replica is being replaced, while in our protocol, the entire root committee is replaced.  Under normal circumstances, this is a rare event as the probability of having more than $c/3-1$ byzantine replicas in root committee is low and once majority of root committee members are honest the protocol runs for sufficient period without changing the root committee.

\subsection{A secret recipe for high scalability and throughput}
The secret for higher scalability and throughput of this protocol can be discovered by plotting values of $n$ against $c$ for an appropriate range of values for $P_f$ in Equation \ref{eqn:basic1}. Figure \ref{nVSc} plots the growth of $c$ against $n$ for $P_f \leq 8.9 \cdot e^{-7}$. We can see that $c$ follows sublinear growth against $n$, such that for very large $n$, $c$ becomes negligible compared to $n$. Thus, as the growth of $c$ slows down compared to the total number of replicas in the network, it also diminishes the effect of quadratic communication ($c^2$) in the root committee on protocol scalability and throughput. We designed our protocol to leverage this property. 

\begin{figure}
    \centering

    \begin{tikzpicture} [scale=0.77]
\begin{semilogyaxis}[
 xlabel={n}, ylabel={c},
    legend style={at={(0.5,-0.15)},
      anchor=north,legend columns=-1},
 ]

 \addplot plot coordinates {
 (10, 5 )
 (40, 18)
 (70, 27)
 (100, 30)
 (130, 33)
 (160,36)
 (190, 36)
 (220, 39)
 (250,39)
 (280,39)
 (310,41)
 (340,42)
 (370,42)
 (400, 42)
 (430, 42)
 (460, 42)
 (490, 42)


 };

 \end{semilogyaxis}
\end{tikzpicture}
  \caption{Growth of $c$ against $n$}
    \label{nVSc}
\end{figure}
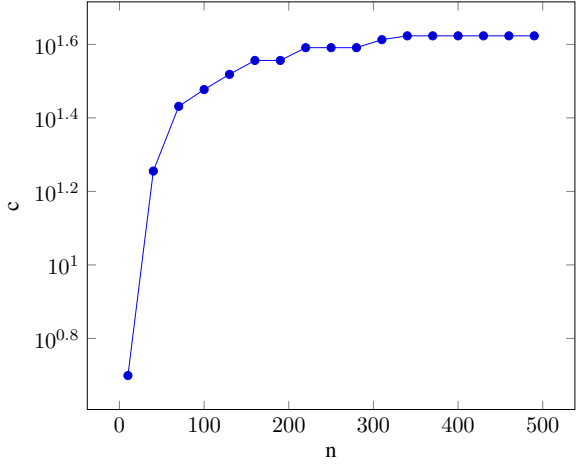

\SetKwFor{Upon}{upon}{do}{end}
\SetKwFor{Check}{check always}
{that}{end}


\begin{algorithm}
\DontPrintSemicolon 
   
\caption{Customized BFT algorithm for root committee member $i$}\label{Algorithm:Customized BFT}

\tcp{primary is a designated node in root committee}
\If{$i$ is primary in root committee} {
Collect transactions and form block\;
Propose block to root committee by broadcasting a pre-prepare message with the block\; 
}

\Upon{receipt of valid pre-prepare message}{
Broadcast prepare message with the proposed block to root committee\;
}
\Upon{receipt of $2c/3+1$ valid prepare messages}{
Broadcast commit message to root committee\;
}

\Upon{receipt of $2c/3+1$ valid commit messages for proposed block}{
Return success and continue to Algorithm \ref{Algorithm:Root committee member}\;
}
\If{received timeoutfailure from root node $j$ for block number $B_t$ }{
\If{$j$ has not previously received from $i$ the most recent missing blocks since $B_t$}{
Update $j$ up to latest block\;
}

}
\If{not returned success by block timeout}{
Broadcast timeoutfailure complaint to root committee\;
Return timeoutfailure and continue with Algorithm \ref{Algorithm:Root committee member}\;
}

\end{algorithm}


\begin{algorithm}
\DontPrintSemicolon 
   
\caption{Root committee member $i$}\label{Algorithm:Root committee member}


\Upon{successful completion of  BFT protocol in Algorithm \ref{Algorithm:Customized BFT}}{
Broadcast the block to regular replicas\;
}
\Upon{$i$ returns timeoutfailure from Algorithm \ref{Algorithm:Customized BFT}}{
Accept messages from root nodes to synchronize local history\;
}

\Upon{receipt of $2f+1$ valid signatures for proposed block hash}{
Commit block\;
Broadcast aggregated $2f+1$ valid signatures\;
}

\Upon{detecting proof of maliciousness: two different signed blocks or receipt of $f+1$ timeout complaints)}{
Broadcast maliciousness proof\;
\tcp{every node uses same random number generation seed}
Randomly select $c$ members of root committee from set of $n$ replicas\;
\eIf{replica $i$ is not member of new root committee}{
Execute Algorithm \ref{Algorithm:Regular ViewChange}\;

}{
Execute Algorithm \ref{Algorithm: root committee ViewChange}
(and concurrently Algorithm \ref{Algorithm:Regular ViewChange} as a regular replica)\;
}

}
\end{algorithm}

\begin{algorithm}[t]
\DontPrintSemicolon

\caption{Regular member replica k}\label{Algorithm:Regular members}

\Upon{receipt of valid block from root committee (until timeout)}{

Sign block hash and broadcast it to root committee\;

}


\If{timeout for current block}{
Broadcast timeout complaint to root committee
}
\Upon{receipt of a message from root committee member with $2f+1$ 
signatures for block}{
Commit block\;
}

\tcp{Initiate view change actions}
\Upon{receipt of a message with $f+1$ timeout complaints or receipt of invalid message (block or signatures)}{ 
\tcp{Transition to new view, based on common random number generation seed}
Randomly select $c$ members of root committee from set of $n$ replicas\;
\eIf{replica $k$ is not member of root committee}{
Execute Algorithm \ref{Algorithm:Regular ViewChange}\;
}{
Execute Algorithm \ref{Algorithm: root committee ViewChange}

}


}

\end{algorithm}

\begin{algorithm}[t]
\DontPrintSemicolon

\caption{ViewChange for regular replica $i$}\label{Algorithm:Regular ViewChange}
Broadcast local history $V_i$ to new root committees\;
\Upon{Receipt of a $Q$ message from a new committee member $j$}{
\If{$2f+1$ replicas are updated to the same history in $Q$}{
Broadcast READY message $R_i$ to new root committee\;
\Upon{Receipt of $P$ from a new root committee member $j$}{
\If{$P$ has at least $2f+1$ distinct READY messages}{
Synchronize local history $V_i$ according to $Q$\;
return to Algorithm \ref{Algorithm:Regular members}\;
}
}
}

}

\end{algorithm}

\begin{algorithm}[t]
\DontPrintSemicolon

\caption{ViewChange for new root committee member $j$}\label{Algorithm: root committee ViewChange}
Broadcast local history $V_j$ to new committee\;
\Upon{Receipt of $V_i$ from replica $i$}{

$Q \gets  Q \cup V_i$\;
\If{$Q$ contains at least $2f+1$ same histories)}{
Broadcast $Q$ to all replicas\;
}

}

\Upon{Receipt of $R_i$}{
$P \gets  P \cup R_i$\;
\If{$P$ has accumulated at least $2f+1$ distinct READY messages}{
Broadcast $P$ to all replicas\;
}

}
Return to Algorithm \ref{Algorithm:Root committee member}
\end{algorithm}

\subsection {Detailed Protocol Operation}
Algorithm \ref{Algorithm:Customized BFT}, Algorithm \ref{Algorithm:Root committee member}, and Algorithm \ref{Algorithm:Regular members} describe the normal execution between root committee and regular replicas. If normal execution fails, then our protocol switches to view change mode executing algorithms \ref{Algorithm:Regular ViewChange} and \ref{Algorithm: root committee ViewChange} to recover from failure.
Note that the root committee members also run the protocols for regular replica nodes in the normal mode.

\paragraph{Normal mode.}
A designated primary node in the root committee proposes a block by broadcasting a pre-prepare message to other root committee members (Algorithm \ref{Algorithm:Customized BFT}, lines 1-4). A pre-prepare message $B=(\langle ``Pre\mbox{-}prepare",v,s,h,d \rangle_p,m)$  contains view number($v$), block sequence number(s), transaction list($m$) and its hash($h$) and previous blockhash($d$).
A replica $i$ in the root committee begins a customized BFT algorithm after receipt of a pre-prepare message. Then, replica $i$ broadcasts a prepare message $\langle ``Prepare",v,s,h,i \rangle_i$
if it finds the pre-prepare message to be valid. The validity check of the pre-prepare message includes checking the validity of $s$, $v$, $d$, $h$ and  transactions inside $m$ (Algorithm \ref{Algorithm:Customized BFT}, lines 5-7).
Upon receipt of more than $2c/3$ prepare messages from other root committee members, the replica $i$ broadcasts a commit message $\langle ``Commit",v,s,h,i \rangle_i$ to the root committee (Algorithm \ref{Algorithm:Customized BFT}, lines 8-10).
If replica $i$ receives more than $2c/3$ commit messages from other root committee members, then it will return success to Algorithm \ref{Algorithm:Root committee member} (Algorithm \ref{Algorithm:Customized BFT}, lines 11-13).


Since our protocol can tolerate up to $2c/3$ Byzantine nodes in the root committee, then the $2c/3+1$  nodes involved above have to include at least one correct replica in root committee during the  consensus process to generate a block. While other correct replicas (at most $c/3-1$) may not be participating in the consensus process as malicious primary may not send them messages. These $c/3-1$ correct replicas will need to sync with other root replicas. Therefore, replica $i$ might receive a \emph{timeoutfailure} for block number $B_t$. If replica $i$ has not received a \emph{timeoutfailure} message from the same replica for the same block or any block with larger sequence number then the replica will forward prepare and commit messages for missing block (Algorithm \ref{Algorithm:Customized BFT}, lines 14-17). If replica $i$ gets timeout and did not receive valid messages during the consensus process, it will return a \emph{tiemoutfailure} error to Algorithm \ref{Algorithm:Root committee member}.

If root committee members successfully generated a proposal block(containing  pre-prepare message, a commit message and $2c/3+1$ signatures  for commit as proof) then they will broadcast it to the regular committee members (Algorithm \ref{Algorithm:Root committee member}, lines 1-3).
Upon receipt of a block from the root committee, regular replicas check if it is signed by at least $2c/3+1$ root committee members and if verified, a regular replica $j$ sends back a signed Approval message $\langle``Approval",v,s,h,j \rangle_j$ 
to the root committee
(Algorithm \ref{Algorithm:Regular members}, lines 1-3). 
Each member of the root committee aggregates $2f+1$ signatures $\sigma(h)$ (including commit signatures from root committee as well as Approval message signatures from regular replicas) and then commits the block and broadcasts Confirm message  $\langle ``Confirm", v,s,h,i, \sigma(h)\rangle_i$ 
to regular replicas (Algorithm \ref{Algorithm:Root committee member}, lines 7-10).
Upon receipt of $2f+1$ valid signatures, regular committee members also commit the block (Algorithm \ref{Algorithm:Regular members}, lines 7-9). But there is a possibility that the root committee might behave in malicious way. 
Furthermore, the primary replica in a root committee might be malicious, thus, even if the number of honest replicas are in the majority, the primary can still cause the root committee to fail. In case a block can't be committed, our protocol switches to view change mode to select an entirely new root committee. 

\paragraph{BFT failure cases.}
 Based on our assumption, there are at least $c/3$ honest replicas in root committee. This results in the following failure cases being relevant:
 \begin{enumerate}
 \item {\em Number of malicious replicas in the root committee is more than $c/3$:}
If the number of malicious replicas is at least $c/3+1$, and they chose to behave maliciously, no block will be generated, as it will not be possible to collect $2c/3+1$ signatures for the block. Thus, regular replicas will timeout and will start the view change process.
  
  
 
 \item {\em Primary in root committee is malicious:} 
 If the primary replica is malicious 
 then the root committee might not be able to generate a block successfully, as the primary might simply not initiate the block proposal. In such a case, a timeout will trigger a view change. 
 
 \item {\em Primary in the root committee as well as at most an additional $(2c/3)-1$ root committee members are malicious:} 
 In this case, the malicious replicas in the root committee can collude with the malicious primary to force honest replicas in the root committee to accept different block proposals. But as honest replicas in the root committee broadcast block proposals (along with root members' commit messages for the block), the other replica will detect this discrepancy and this will trigger a view change. 
 \end{enumerate}
 
More details of these cases and others are addressed in Section~\ref{Proof of correctness}.

\paragraph{View change.}
Unlike ordinary BFT protocols, a view change in our protocol is achieved by replacing the root committee members with $c$ new members. A view change occurs if the root committee fails to generate a new valid block. If no new block is generated, then after a timeout a new committee immediately takes over and continues the consensus process. During each epoch, a regular replica waits to receive a proposed block from the root committee. If a regular replica $i$ does not receive the block after a timeout then it considers that the root committee has failed and reports this to the root committee.
If $f+1$ nodes report a timeout, then this triggers the view change process in root committee and the root committee members forward $f+1$ timeout reports to regular replicas triggering view change in regular replicas. In the view change(triggered in Algorithm \ref{Algorithm:Root committee member}, lines 11-19 and Algorithm \ref{Algorithm:Regular members}, lines 10-17), each node selects another set of $c$ replicas for the new root committee (using a pre-specified random seed, which guarantees that every replica selects the same root committee). Each replica broadcasts its local history $V_i=\langle``ViewChange",v+1,s',H,i,(\sigma (H)) \rangle_i $ to the new root committee members. Local history includes the latest block sequence number($s'$), its hash($H$), incremented view number($v+1$), and signature evidence of at least $2f+1$ ($\sigma (H)$) replicas that have approved the block.

The new root committee members wait to receive $2f+1$ local histories $V_i$ from all replicas and aggregate them into $Q$. Once it receives $V_i$ from replica $i$, its local history $H_i$ (latest committed block) is extracted from $V_i$. Out of $2f+1$, it is guaranteed, that at least $f+1$ are honest replicas.
Thus, the most recent history in these $f+1$ replicas will match and this will be the starting point for the next block to be generated. Among other information, $V_i$ also contains 
the latest block sequence number (height of the block) and the replica $id$ that proposed the block.
Root committee members will broadcast $Q$ to all replicas(Algorithm \ref{Algorithm: root committee ViewChange}). Upon receipt of $Q$  replica $i$ makes sure that its history matches with history of $Q$(agreed by at least $f+1$ replicas). If its history matches, replica $i$ sends back Ready message ($R_i$) to new root committee. Root committee member $j$ will aggregate $2f+1$ $P$ messages and broadcast it to all replicas. Upon receipt of $P$ that includes $2f+1$ ready messages, replica $i$ is now ready to take part in new view.  If replica $i$'s history does not match that of $Q$ it will synchronize its history(Algorithm \ref{Algorithm:Regular ViewChange}). 

Since it is not guaranteed that at least $2c/3+1$ replicas in the root committee will be honest, it is difficult to transfer the transactions that have been previously agreed in the previous view by $2c/3+1$ root members during the preparation phase.
To guarantee all transactions are executed properly, clients can resend their transactions to the new primary of the new root committee, after they realize that they have not received any response for their transaction. 
\section{Proof of Correctness } 
\label{Proof of correctness}

We prove that the algorithm either produces a block in an epoch or view change occurs.









\begin{lemma}
\label{lemma:block receipt guarantee}
If a block is signed by $2c/3+1$ root committee nodes then it is guaranteed
that all correct replicas will receive it.
\end{lemma}

\begin{proof}
Since the root committee consists of $c/3$ correct nodes,
then one of the $2c/3 + 1$ nodes that signed the block has to be 
a correct node.
That correct node will broadcast the block to all replicas. 
\end{proof}

\begin{lemma}
\label{lemma:block commit}
If at least $2f+1$ correct replicas agree on the same block then it gets committed to local histories of all correct nodes, otherwise a view change will occur.
\end{lemma}

\begin{proof}
From Lemma \ref{lemma:block receipt guarantee},
if a block is signed by at least $2c/3 + 1$ committee members then it is broadcast to all 
replicas.

In order for a block to be committed it must be signed by at least $2f+1$ regular replicas (Algorithm \ref{Algorithm:Root committee member}, lines 7-10). Otherwise, 
if this does not happen then there are two possible scenarios:
\begin{itemize}
    \item
    Two different blocks have been proposed by the root committee $C$. This can happen if $2c/3$ Byzantine nodes from $C$ and one additional correct committee node $u \in C$ decide one block, while again $2c/3$ Byzantine and one correct node $v \in C$, $v \neq u$, decide on another block. The blocks may be sent to different replica sets, or replicas may pick one of the two to sign. In either case, some correct node of the root committee receives two different  signed block hashes from the replicas (\ref{Algorithm:Root committee member}, line 11).
    \item
    $f+1$ replicas have not received any valid block, and therefore $f+1$ timeout complaint messages are reported to the root committee (\ref{Algorithm:Root committee member}, line 11).
    
\end{itemize}
Both of these cases will create a proof of maliciousness which will trigger a view change (Algorithm \ref{Algorithm:Root committee member}, lines 12-18 and Algorithm \ref{Algorithm:Regular members} line 10.).

Therefore, if view change does not occur then a block was signed by $2f+1$ regular replicas.
This causes the block to be committed by the correct replicas 
in root committee 
(\ref{Algorithm:Root committee member}, lines 7-9)
and also by the correct regular replicas 
(Algorithm \ref{Algorithm:Regular members} line 7-9).
\end{proof}

\begin{lemma}
\label{lemma:consistent histories}
After a view change completes, all histories of correct nodes are consistent.
\end{lemma}

\begin{proof}
During a view change each replica sends its local history $V_i$ to the new root committee
(First line in Algorithms \ref{Algorithm:Regular ViewChange} and \ref{Algorithm: root committee ViewChange}).
Each member of the root committee $j$ collects at least $2f+1$ local histories and aggregates them into a single message $Q$ which it broadcasts to all replicas
(Algorithm \ref{Algorithm: root committee ViewChange}, lines 2-7).
Each replica $i$ receives the $Q$ messages from all members of the new root committee,
and verifies that at least $2f+1$ replicas have agreed on the same history and most recent block,
and then $i$ broadcasts the READY message to the new root committee (Algorithm \ref{Algorithm:Regular ViewChange}, lines 3-4).   
Each node $j$ of the new root committee accumulates in $P$ the READY messages that it receives, and when it receives at least $2f+1$ messages then it broadcasts $P$
(Algorithm \ref{Algorithm: root committee ViewChange}, lines 8-12).
When the regular replica $i$ receives message $P$ with at least $2f+1$ READY messages, 
then it updates its local history (Algorithm \ref{Algorithm:Regular ViewChange}, lines 5-10).

Suppose that two different histories $H_1 \in V_{j_1}$ and $H_2 \in V_{j_2}$ appear in $Q_{j_1}$ and $Q_{j_2}$, respectively,
where $j_1, j_2 \in C$.
Since each of $Q_{j_1}$ and $Q_{j_2}$ are constructed from $2f+1$ replicas' common histories,
then it has to be that $f+1$ are common nodes for $H_1$ and $H_2$.
Since there are at most $f$ Byzantine nodes, there is a common correct node that has 
proposed two different histories $H_1$ and $H_2$, which is impossible.
Thus, each correct node will update its local history to the agreed history of at least $2f+1$ replicas. 
\end{proof}

From Lemmas \ref{lemma:block commit} and \ref{lemma:consistent histories} we obtain the following theorem.

\begin{theorem}
Under normal operation and after a view change, all correct nodes maintain a consistent history
which includes the latest committed block.
\end{theorem}
\begin{figure*}
\begin{tikzpicture}
\begin{groupplot}[group style={group size= 3 by 1},height=6.5cm,width=6.5cm,legend entries = {Bchain-3,BFT,EfficientBFT}, legend style={at={(0.61,.95)},anchor=north east}]
\nextgroupplot[title=Block size 5k ,ylabel={Latency in secs },xlabel=Network size, ymin=0,ymax=20,
symbolic x coords={10,40,70,100,130,200},xtick=data
], 
\node [text width=1em,anchor=north west] at (rel axis cs: 0.8,1)  {\subcaption{\label{a}}};

 \addplot plot coordinates {

 (40, 2.0)
 (70, 3.2)
 (100, 4.9)
 (130, 6.1)
 (200,12.4)

 };
\addplot plot coordinates {

 (40, 1.8)
 (70, 3.1)
 (100, 4.7)
 (130, 6.4)
 (200,12)
 };
\addplot plot coordinates {

 (40, 1.4)
 (70, 2.16)
 (100, 2.9)
 (130, 3.5)
 (200,4.6)
 };
 \legend{$Bchain-3$\\$PBFT$\\$Proteus$\\} 
\nextgroupplot[title=Block size 10k ,
xlabel=Network size, ymin=0,ymax=30,
symbolic x coords={10,40,70,100,130,200},xtick=data],
\node [text width=1em,anchor=north west] at (rel axis cs: 0.8,1)  {\subcaption{\label{b}}};
\addplot plot coordinates {
 (40, 3.6)
 (70, 5.7)
 (100, 8.46)
 (130, 10.37)
 (200,19.5)

 };
\addplot plot coordinates {

 (40, 3.7)
 (70, 6.63)
 (100, 10.6)
 (130, 15.6)
 (200,25.4)
 };
\addplot plot coordinates {

 (40, 2.5)
 (70, 4)
 (100, 5.5)
 (130, 6.6)
 (200,8.9)
 };
 \legend{$Bchain-3$\\$PBFT$\\$Proteus$\\}  
 \nextgroupplot[title=Block size 15k,ymin=0 ,ymax=40
 ,
 xlabel=Network size, 
 symbolic x coords={10,40,70,100,130,200},xtick=data],
 \node [text width=1em,anchor=north west] at (rel axis cs: 0.8,1)  {\subcaption{\label{C}}};
 \addplot plot coordinates {

 (40, 5.3)
 (70, 8.16)
 (100, 12)
 (130, 14.4)
 (200,26.7)

 };
\addplot plot coordinates {

 (40, 6.14)
 (70, 9.8)
 (100, 16.64)
 (130, 24)
 (200,35.5)
 };
\addplot plot coordinates {

 (40, 3.7)
 (70, 6.2)
 (100, 8.2)
 (130, 9.6)  
 (200,13.4)
 };
 \legend{$Bchain-3$\\$PBFT$\\$Proteus$\\} 
\end{groupplot}

 \end{tikzpicture}
 \caption{Protocol latencies with block sizes 5000, 10000 and 15000} \label{fig:Block-Latency-Measurement}

\end{figure*}
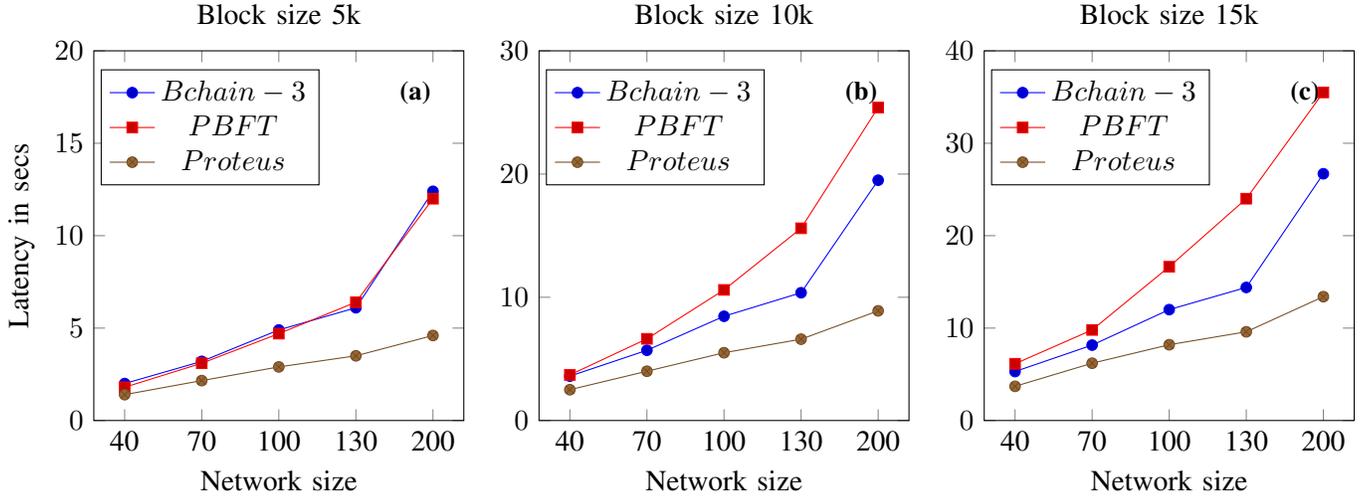

\begin{figure*}
    \centering
   
\begin{tikzpicture}

    \begin{groupplot}[group style={group size= 3 by 1,horizontal sep=18pt},ybar,height=6cm,width=6.5cm, enlargelimits=0.13,legend entries = {Bchain-3,PBFT,Proteus},/pgf/bar width=4.5pt,
    grid=both,     minor tick num=2, scaled ticks=base 10:-3,    
]
\nextgroupplot[title= Block size 5k ,ylabel={Throughput in tx/secs },symbolic x coords={10,40,70,100,130,200},xtick=data,xlabel = Network size, ymax=4000],
\node [text width=1em,anchor=north west] at (rel axis cs: 0,1)  {\subcaption{\label{a}}};

\addplot coordinates {(40,2454) (70,1558) (100,1015) (130,815) (200,402)};


\addplot coordinates {(40,2807) (70,1584) (100,1054) (130,785)(200,418)};


\addplot coordinates {(40,3546) (70,2309) (100,1720) (130,1444)(200,1086)};

\legend{Bchain-3,PBFT,Proteus}
\nextgroupplot[title= Block size 10k ,
symbolic x coords={10,40,70,100,130,200},xtick=data, xlabel = Network size, ymax=4000],
\node [text width=1em,anchor=north west] at (rel axis cs: 0,1)  {\subcaption{\label{b}}};

\addplot coordinates {(40,2769) (70,1783) (100,1181) (130,964) (200,513)};

\addplot coordinates { (40,2704) (70,1508) (100,945) (130,666) (200, 394)};

\addplot coordinates { (40,3921) (70,2491) (100,1811) (130,1516)  (200,1123)};

\nextgroupplot[title= Block size 15k ,
,symbolic x coords={10,40,70,100,130,200},xtick=data, xlabel = Network size],
\node [text width=1em,anchor=north west] at (rel axis cs: 0,1)  {\subcaption{\label{c}}};

\addplot coordinates { (40,2852) (70,1836) (100,1255) (130,1044)(200,560) };

\addplot coordinates {(40,2644) (70,1534) (100,959) (130,630) (200,422)};

\addplot coordinates {(40,4029) (70,2433) (100,1838) (130,1555) (200,1119)};

\end{groupplot}
\end{tikzpicture}
\caption{Protocol throughput with Block size 5000, 10000 and 15000} \label{fig:Block-Throughput}

\end{figure*}
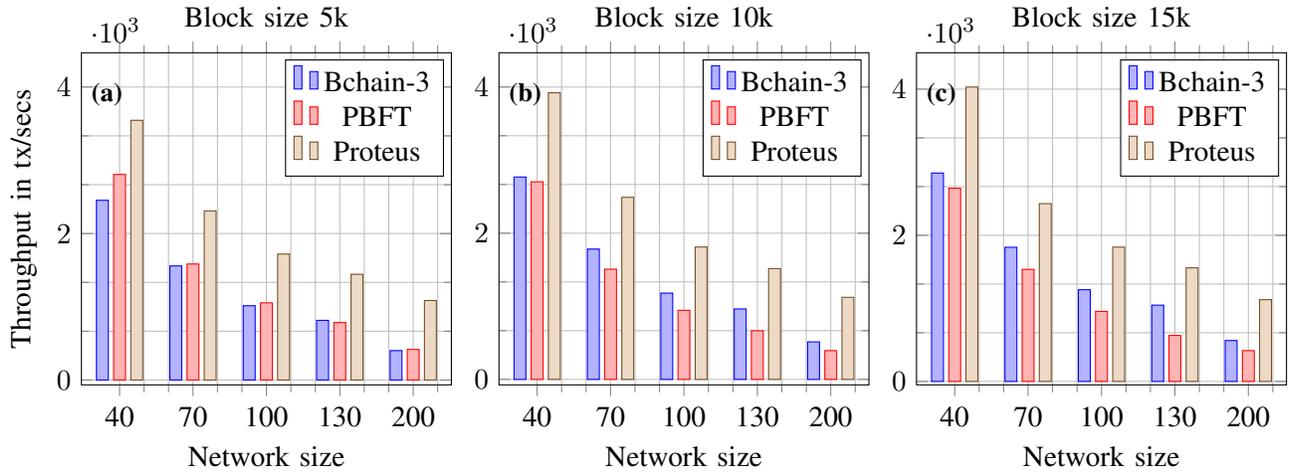

\section{Message Complexity Analysis}
\label{Message Complexity}

The root committee nodes execute a BFT protocol among them with quadratic message complexity with respect to the committee $c$, where both the prepare and commit phases cause $O(c^2)$ messages. The communication between the root committee $C$ and regular replicas requires $O(cn)$ messages: the $c$ root committee nodes send the block to $n$ regular replicas ($cn$ messages),
and then the root committee nodes receive back regular replicas signatures ($cn$ messages)
and upon aggregation of signatures root committee members broadcast at least $2f+1$ aggregated signatures to regular replicas ($cn$ messages). 
Thus, total message complexity during normal operation will be
$O(c^2+cn)$
and for large $n$, we have $n>>c$, we will get:
\begin{equation*}
\mbox{Message complexity ($n>>c$): }  O(cn)
\end{equation*}
View change is either caused by timeout or by receiving proof of maliciousness against the root committee. Based on Algorithm \ref{Algorithm:Regular ViewChange} and Algorithm \ref{Algorithm: root committee ViewChange} there are multiple rounds of communication between  root committee of size $c$ and the regular replicas of the size $n-c$ effectively making view change message complexity as $O(cn)$.

\section{Experiments and Evaluations}
\label{section:experiments}

We have implemented the Proteus protocol in about 2.5k lines of Golang code. We also implemented PBFT to be used as a baseline. We selected PBFT because while other protocols have improved on PBFT performance, they usually switch to $n \times n$ broadcast (legacy PBFT) as a fallback measure if a certain threshold number of failures have occurred. 
Additionally we also implemented Bchain-3 which belongs to the family of chain-based BFT protocols, where instead of a broadcast the protocols propagate messages along a predetermined chain order. 
Bchain-3 has the potential to decrease the number of messages that are sent in the network, however, latency may increase due to a possibly long chain of message relays.
Both BFT and Bchain-3 are also implemented in Golang for the experiments. We think comparing our protocol with different flavors of BFT-based protocols (broadcast and chain based BFT) provide good insight into the performance of Proteus.

We conducted our experiments in the Amazon Web Services (AWS) cloud. For each replica in the network we used instances of type {\em t2.large} in AWS. Each {\em t2.large} instance has 2 virtual CPU cores with 8GB of memory. The experiments were performed with network sizes of  $40$, $70$, $100$ and $130$ and $200$ replicas. We used Equation \ref{eqn:basic1} to get suitable value for the root committee size $c$,
for the different numbers of replicas. 
Given $n$ and maximum failure probability of $P_f\leq 8.9 \cdot 10^{-7}$,
we selected the various root committee sizes to be  18 (40), 27 (70), 30 (100), 33 (130), and 36(200). We also used different block sizes with $5k$, $10k$ and  $15k$ transactions. The transactions in the block are simple blockchain transactions that are randomly generated and transfer funds from one account to another account. 

The other messages to obtain consensus are smaller size since they contain only hashes of the block, signatures and other information that do not take as much space.
Replicas perform regular operations on the blocks they receive and
they also maintain a complete history of blocks (complete blockchain).
Namely, operations performed by each replica include checking its history, hashes, and signatures to verify the validation of transactions in the block.

\subsection{Analysis of experimental results}
We measure latency in an epoch as the time between the primary root committee proposing a block until the time that the block is inserted in all the local histories.
We average the latency over several epochs.
Figure \ref{fig:Block-Latency-Measurement} provides a comparison of latency measurements among Bchain-3, PBFT, and Proteus. 
The difference in latency is small for smaller network sizes, but it increases as the network size grows. Additionally, we observe that the latency is affected by the block size, so that if the block size increases then the latency increases as well. For example in Figure \ref{fig:Block-Latency-Measurement} (a), the latency difference among protocols is smaller, but as the block size increases (Figure \ref{fig:Block-Latency-Measurement} (b) and (c))
we can see that the latency of PBFT increases much faster (due to its high message complexity) followed by Bchain-3 (due to longer length of the critical chain path). In all cases Proteus provides better latency than PBFT and Bchain-3, which demonstrates the scalability of our protocol.

We also measure the throughput, which is the number of committed transactions per second that are appended to the blockchain (Figure \ref{fig:Block-Throughput}).
Similarly to latency, by observing Figure \ref{fig:Block-Throughput} we can see that Proteus outperforms PBFT and Bchain-3 for all the test cases. The performance superiority of Proteus is more visible (two times better than other two protocols), 
when network size increases. For network size $200$ and blocksize $15k$ throughput for Bchain, PBFT and Proteus are $560$tx/sec,$422$ tx/sec, and $1119$ tx/sec. For size $10k$ throughput for Bchain, PBFT and Proteus is $513$ tx/sec,$394$ tx/sec, and $1123$ tx/sec respectively.  For blocksize $5k$, throughput for Bchain, PBFT and Proteus is $402$ tx/sec, $418$ tx/sec and $1086$ respectively.
This also demonstrates the scalability potential of our protocol.

\section{Conclusion and Future Work}
In this paper we presented Proteus, a BFT-based consensus protocol for blockchains that provides better latency and throughput than the state of the art BFT protocols (Bchain and PBFT). Proteus, provides consistent performance regardless of number of failures encountered in the network whereas other BFT  improvements suffer from fall back performance degradation as the number of failures in the network reach the threshold.
Our future work will mainly focus on how to further improve scalability(through sharding) and security of BFT-based protocols while having minimum effects on throughput and latency. 
\label{Conclusions}

\bibliographystyle{IEEEtran}
\bibliography{OPBFT,TBFT}

\begin{thebibliography}{10}
\providecommand{\url}[1]{#1}
\csname url@samestyle\endcsname
\providecommand{\newblock}{\relax}
\providecommand{\bibinfo}[2]{#2}
\providecommand{\BIBentrySTDinterwordspacing}{\spaceskip=0pt\relax}
\providecommand{\BIBentryALTinterwordstretchfactor}{4}
\providecommand{\BIBentryALTinterwordspacing}{\spaceskip=\fontdimen2\font plus
\BIBentryALTinterwordstretchfactor\fontdimen3\font minus
  \fontdimen4\font\relax}
\providecommand{\BIBforeignlanguage}[2]{{%
\expandafter\ifx\csname l@#1\endcsname\relax
\typeout{** WARNING: IEEEtran.bst: No hyphenation pattern has been}%
\typeout{** loaded for the language `#1'. Using the pattern for}%
\typeout{** the default language instead.}%
\else
\language=\csname l@#1\endcsname
\fi
#2}}
\providecommand{\BIBdecl}{\relax}
\BIBdecl

\bibitem{Jala1807:Window}
M.~Jalalzai and C.~Busch, ``Window based {BFT} blockchain consensus,'' in
  \emph{2018 IEEE International Conference on Blockchain (Blockchain 2018)},
  Halifax, Canada, Jul. 2018.

\bibitem{DBLP:journals/corr/LiuLKA16a}
\BIBentryALTinterwordspacing
J.~Liu, W.~Li, G.~O. Karame, and N.~Asokan, ``Scalable byzantine consensus via
  hardware-assisted secret sharing,'' \emph{CoRR}, vol. abs/1612.04997, 2016.
  [Online]. Available: \url{http://arxiv.org/abs/1612.04997}
\BIBentrySTDinterwordspacing

\bibitem{Guerraoui:2010:NBP:1755913.1755950}
\BIBentryALTinterwordspacing
R.~Guerraoui, N.~Kne\v{z}evi\'{c}, V.~Qu{\'e}ma, and M.~Vukoli\'{c}, ``The next
  700 bft protocols,'' in \emph{Proceedings of the 5th European Conference on
  Computer Systems}, ser. EuroSys '10.\hskip 1em plus 0.5em minus 0.4em\relax
  New York, NY, USA: ACM, 2010, pp. 363--376. [Online]. Available:
  \url{http://doi.acm.org/10.1145/1755913.1755950}
\BIBentrySTDinterwordspacing

\bibitem{SBFT}
G.~Golan{-}Gueta, I.~Abraham, S.~Grossman, D.~Malkhi, B.~Pinkas, M.~K. Reiter,
  D.~Seredinschi, O.~Tamir, and A.~Tomescu, ``{SBFT:} a scalable decentralized
  trust infrastructure for blockchains,'' \emph{CoRR}, vol. abs/1804.01626,
  2018.

\bibitem{194906}
I.~Eyal, A.~E. Gencer, E.~G. Sirer, and R.~V. Renesse, ``Bitcoin-ng: A scalable
  blockchain protocol,'' in \emph{13th {USENIX} Symposium on Networked Systems
  Design and Implementation ({NSDI} 16)}.\hskip 1em plus 0.5em minus
  0.4em\relax Santa Clara, CA: {USENIX} Association, 2016, pp. 45--59.

\bibitem{Nakamoto_bitcoin:a}
S.~Nakamoto, ``Bitcoin: A peer-to-peer electronic cash system,
  http://bitcoin.org/bitcoin.pdf.''

\bibitem{DBLP:conf/ifip114/Vukolic15}
M.~Vukolic, ``The quest for scalable blockchain fabric: Proof-of-work vs. {BFT}
  replication,'' in \emph{Open Problems in Network Security - {IFIP} {WG} 11.4
  International Workshop, iNetSec 2015, Zurich, Switzerland, October 29, 2015,
  Revised Selected Papers}, 2015, pp. 112--125.

\bibitem{6912770}
K.~J. O'Dwyer and D.~Malone, ``Bitcoin mining and its energy footprint,'' in
  \emph{25th IET Irish Signals Systems Conference 2014 and 2014 China-Ireland
  International Conference on Information and Communications Technologies (ISSC
  2014/CIICT 2014)}, June 2014, pp. 280--285.

\bibitem{IsBitCoinDecentralized}
A.~Gervais, G.~O. Karame, V.~Capkun, and S.~Capkun, ``Is bitcoin a
  decentralized currency?'' \emph{IEEE Security Privacy}, vol.~12, no.~3, pp.
  54--60, May 2014.

\bibitem{Sok}
\BIBentryALTinterwordspacing
J.~Bonneau, A.~Miller, J.~Clark, A.~Narayanan, J.~A. Kroll, and E.~W. Felten,
  ``Sok: Research perspectives and challenges for bitcoin and
  cryptocurrencies,'' in \emph{Proceedings of the 2015 IEEE Symposium on
  Security and Privacy}, ser. SP '15.\hskip 1em plus 0.5em minus 0.4em\relax
  Washington, DC, USA: IEEE Computer Society, 2015, pp. 104--121. [Online].
  Available: \url{https://doi.org/10.1109/SP.2015.14}
\BIBentrySTDinterwordspacing

\bibitem{Ethereum-EIP-150}
\BIBentryALTinterwordspacing
D.~G. WOOD, ``Ethereum: A secure decentralised generalised transaction
  ledger,'' pp. 1--33, 2017. [Online]. Available:
  \url{https://ethereum.github.io/yellowpaper/paper.pdf}
\BIBentrySTDinterwordspacing

\bibitem{Smart-pools}
\BIBentryALTinterwordspacing
L.~Luu, Y.~Velner, J.~Teutsch, and P.~Saxena, ``Smartpool: Practical
  decentralized pooled mining,'' in \emph{26th {USENIX} Security Symposium
  ({USENIX} Security 17)}.\hskip 1em plus 0.5em minus 0.4em\relax Vancouver,
  BC: {USENIX} Association, 2017, pp. 1409--1426. [Online]. Available:
  \url{https://www.usenix.org/conference/usenixsecurity17/technical-sessions/presentation/luu}
\BIBentrySTDinterwordspacing

\bibitem{wust2016securityofblockchaintech}
\BIBentryALTinterwordspacing
K.~W{\"u}st, ``Security of blockchain technologies,'' Jul 2016, accessed:
  2016-11-08. [Online]. Available:
  \url{http://e-collection.library.ethz.ch/eserv/eth:49632/eth-49632-01.pdf}
\BIBentrySTDinterwordspacing

\bibitem{Lamport:1984:UTI:2993.2994}
\BIBentryALTinterwordspacing
L.~Lamport, ``Using time instead of timeout for fault-tolerant distributed
  systems.'' \emph{ACM Trans. Program. Lang. Syst.}, vol.~6, no.~2, pp.
  254--280, Apr. 1984. [Online]. Available:
  \url{http://doi.acm.org/10.1145/2993.2994}
\BIBentrySTDinterwordspacing

\bibitem{Kotla:2008:ZSB:1400214.1400236}
\BIBentryALTinterwordspacing
R.~Kotla, A.~Clement, E.~Wong, L.~Alvisi, and M.~Dahlin, ``Zyzzyva: Speculative
  byzantine fault tolerance,'' \emph{Commun. ACM}, vol.~51, no.~11, pp. 86--95,
  Nov. 2008. [Online]. Available:
  \url{http://doi.acm.org/10.1145/1400214.1400236}
\BIBentrySTDinterwordspacing

\bibitem{Castro:1999:PBF:296806.296824}
\BIBentryALTinterwordspacing
M.~Castro and B.~Liskov, ``Practical byzantine fault tolerance,'' in
  \emph{Proceedings of the Third Symposium on Operating Systems Design and
  Implementation}, ser. OSDI '99.\hskip 1em plus 0.5em minus 0.4em\relax
  Berkeley, CA, USA: USENIX Association, 1999, pp. 173--186. [Online].
  Available: \url{http://dl.acm.org/citation.cfm?id=296806.296824}
\BIBentrySTDinterwordspacing

\bibitem{Luu:2016:SSP:2976749.2978389}
\BIBentryALTinterwordspacing
L.~Luu, V.~Narayanan, C.~Zheng, K.~Baweja, S.~Gilbert, and P.~Saxena, ``A
  secure sharding protocol for open blockchains,'' in \emph{Proceedings of the
  2016 ACM SIGSAC Conference on Computer and Communications Security}, ser. CCS
  '16.\hskip 1em plus 0.5em minus 0.4em\relax New York, NY, USA: ACM, 2016, pp.
  17--30. [Online]. Available: \url{http://doi.acm.org/10.1145/2976749.2978389}
\BIBentrySTDinterwordspacing

\bibitem{Fischer:1985:IDC:3149.214121}
\BIBentryALTinterwordspacing
M.~J. Fischer, N.~A. Lynch, and M.~S. Paterson, ``Impossibility of distributed
  consensus with one faulty process,'' \emph{J. ACM}, vol.~32, no.~2, pp.
  374--382, Apr. 1985. [Online]. Available:
  \url{http://doi.acm.org/10.1145/3149.214121}
\BIBentrySTDinterwordspacing

\bibitem{10.1007/978-3-319-14472-6_7}
S.~Duan, H.~Meling, S.~Peisert, and H.~Zhang, ``Bchain: Byzantine replication
  with high throughput and embedded reconfiguration,'' in \emph{Principles of
  Distributed Systems}, M.~K. Aguilera, L.~Querzoni, and M.~Shapiro, Eds.\hskip
  1em plus 0.5em minus 0.4em\relax Cham: Springer International Publishing,
  2014, pp. 91--106.

\bibitem{Hot-stuff}
D.~Malkhi, ``Blockchain in the lens of {BFT}.''\hskip 1em plus 0.5em minus
  0.4em\relax Boston, MA: {USENIX} Association, 2018.

\bibitem{tendermint}
\BIBentryALTinterwordspacing
E.~Buchman, ``Tendermint: Byzantine fault tolerance in the age of
  blockchains,'' Jun 2016, accessed: 2017-02-06. [Online]. Available:
  \url{http://atrium.lib.uoguelph.ca/xmlui/bitstream/handle/10214/9769/Buchman_Ethan_201606_MAsc.pdf}
\BIBentrySTDinterwordspacing

\bibitem{BlockchaininDWild}
C.~Cachin and M.~Vukolic, ``Blockchain consensus protocols in the wild (keynote
  talk),'' in \emph{{DISC}}, ser. LIPIcs, vol.~91.\hskip 1em plus 0.5em minus
  0.4em\relax Schloss Dagstuhl - Leibniz-Zentrum fuer Informatik, 2017, pp.
  1:1--1:16.

\end{thebibliography}

\end{document}